\documentclass[11pt,a4paper]{article}
\usepackage{amsmath,amssymb,amsfonts,amsthm}
\usepackage[pdftex]{graphicx, color}

\newcommand{\E}{{\bf{E}}}

\newcommand{\PP}{{\bf{P}}}
\newcommand{\Var}{{\bf{Var}}}

\newtheorem{lem}{Lemma}

\long\def\symbolfootnote[#1]#2
{\begingroup%
\def\thefootnote{\fnsymbol{footnote}}\footnote[#1]{#2}\endgroup}

\begin{document}

\bibliographystyle{plain}
\parindent=0pt
\centerline{\LARGE \bfseries Preferred attachment model}

\medskip

\centerline{\LARGE \bfseries of affiliation network}

\bigskip

\centerline{ Mindaugas Bloznelis \ \ and \ \  Friedrich G\"otze}

\bigskip

\bigskip

\quad  \qquad  \qquad      \qquad  \qquad      \qquad   Vilnius University    \qquad      \qquad        Bielefeld University
    
\quad  \qquad  \qquad      \qquad  \qquad      \qquad   LT-03225 Vilnius     \qquad      \qquad     \ \    D-33501 Bielefeld    

\quad  \qquad  \qquad      \qquad  \qquad      \qquad   Lithuania     \qquad      \qquad   \qquad    \qquad   Germany

\par\vskip 3.5em

\begin{abstract}
In an affiliation network 
 vertices are linked to attributes  and two vertices are declared adjacent 
whenever they share a common attribute. For example, two customers of an internet shop
are called adjacent if they have bought the same or similar items. 
Assuming that each
%Here the preferred attachment 
%principle means that a 
newly arrived customer is linked preferentially to already popular 
items we obtain
%We show that 
a preferred attachment model of an evolving affiliation network.
We show that the network 
has a scale-free property and establish the asymptotic degree distribution. 
\par\end{abstract}

\section{Introduction and results}

A preferential attachment model of evolving network assumes that each newly arrived vertex
is attached preferentially to already well connected
sites, \cite{Barabasi}. The preferential attachment principle is usually realised by
setting the probability of a link between the new vertex $v'$ and an old vertex $v$ to 
be an increasing function of the degree of $v$ (the number of neighbours of $v$).
This scheme can be adapted to affiliation networks. In an affiliation network 
 vertices are linked to attributes  and two vertices are declared adjacent 
whenever they share a common attribute. For example two customers of an internet shop
are called adjacent if they have bought the same or similar items. 
Here the preferred attachment 
principle means that a newly arrived customer is linked preferentially to already highly popular 
items, thus, further increasing their popularity.
 In the present study we show that a preferred attachment model of an affiliation network
has a scale-free property and establish the asymptotic degree distribution.

{\bf Model.} Given $\lambda>0$ and 
integer $k>0$, let $l\ge 0$ be an integer such that $\lambda\le k+l$.
Consider an internet library which contains $w_1,\dots, w_l$ books/items at the beginning.
Every book $w_j$ is prescribed initial score $s(w_j)=1$. On the first step new books
$w_{l+1},\dots, w_{l+k}$ arrive to the library, each having initial score $1$. Then the 
first customer $v_1$ visits the library and downloads  books independently at random:
a book $w$ is chosen with probability $p_{1,s(w)}=\lambda s(w)(l+k)^{-1}$.
Every book chosen by $v_1$ increases its score by one.

The collection of books of the library after  $n$ steps is denoted 
$W_n=\{w_1,\dots,w_{l+nk}\}$.
%The set of visitors is denoted $V_n=\{v_1,\dots, v_n\}$.
On the $n+1$th step $k$ new books arrive to the library, each having initial score $1$. 
Then the 
customer $v_{n+1}$ enters the library  and downloads  books of the library independently 
at random:
a book $w$ is downloaded with probability 
\begin{displaymath}
 p_{n+1,s(w)}=\lambda s(w)(l+(n+1)k+n\lambda)^{-1}
\end{displaymath}
proportional to the score $s(w)$ of $w$.
Here $s(w)-1$ is the number of vertices from $V_n=\{v_1,\dots, v_n\}$ 
that have downloaded the book $w$. 
Every book chosen by $v_{n+1}$ increases its score by one. 

We may  interpret  books as bins. Each newly arrived bin contains a single ball. 
A new customer $v_{n+1}$
throws balls into  
bins $w_1,\dots, w_{l+(n+1)k}$ at random:
each bin $w$ receives a ball with probability $p_{n+1,s(w)}$ 
and independently of the other bins.
The score $s(w)$ counts the (current) number of balls in the bin $w$. 
This number may increase with $n$. 
It measures the popularity (attractiveness) of the book $w$. Hence, 
popular books  have  higher chances to be chosen.

We call customers $v_s$ and $v_t$ adjacent if some book has been downloaded by both of them.
We are interested in the graph $G_n$ on the vertex set $V_n$ 
defined by this 
adjacency relation.

\medskip

{\bf Results.} In the present note we address the question about 
the degree sequence of $G_n$.
We shall show that for every $i=0,1,\dots$, the number of vertices $v\in V_n$
 of $G_n$ having degree
$d(v)=i$  converges to a limit and identify this limit. Namely, we have
as $n\to+\infty$ 
\begin{eqnarray}\nonumber
\frac{\#\{v\in V_n: \, d(v)=i\}}{n}
&\to&
(1+\alpha)
\E {\mathbb I}_{\{Z\le i,\Lambda\ge 1 \}}
\frac{\Gamma(i+2\Lambda)}{\Gamma(i+2\Lambda+\alpha+2)}
\frac{\Gamma(Z+2\Lambda+\alpha+1)}{\Gamma(Z+2\Lambda)},
\quad
i\ge 1,
\\
\label{2014-01-08x1}
\frac{\#\{v\in V_n: \, d(v)=0\}}{n}
&\to&
\E{\mathbb I}_{\{Z=0\}}\frac{1+\alpha}{2\Lambda+1+\alpha}.
\end{eqnarray}
Here $\alpha=k/\lambda$.
$\Gamma$ denotes Euler's Gamma function.
 $\Lambda$ denotes a Poisson random variable 
with mean $\lambda$. $Z$ is a compound Poisson random variable 
\begin{equation}
 Z=\sum_{i=1}^{\Lambda}T_i,
\end{equation}
where  $T_1, T_2,\dots$ are independent random variables independent of $\Lambda$ and 
having the same probability 
distribution 
\begin{equation}\label{2014-01-08x2}
\PP(T_1=j)=x_{j+1},
\qquad
x_{j+1}=(1+\alpha)\Gamma(2+\alpha)\frac{\Gamma(j+1)}{\Gamma(3+\alpha+j)},
\qquad
j=0,1,2,\dots.
\end{equation}
From (\ref{2014-01-08x1}) we  find the tail behaviour of the limiting degree 
distribution. Let $y_i$ denote the quantity on the right hand side of (\ref{2014-01-08x1}).
We have as $i\to+\infty$
\begin{equation}\label{2014-01-13+1}
y_i\sim \lambda(1+\alpha)^2\Gamma(2+\alpha)i^{-2-\alpha}\ln i.
\end{equation}
Here and below we write $z_i\sim q_i$ whenever $z_i/q_i\to 1$ as $i\to+\infty$.

Numbers $x_i$ have interesting interpretation. They are
 limits of the fractions of the number of books having score $i$:
\begin{equation}\label{2014-01-09x}
\lim_{n\to+\infty}\frac{\#\{w\in W_n: \, s(w)=i\}}{nk}
=  
x_i,
\qquad
i=1,2,\dots.
\end{equation}
From the properties of Gamma function 
(formula (6.1.46) of \cite{AS}) we conclude that  
 the sequence
$\{x_i\}_{i\ge 1}$ obeys a power law with exponent $2+\alpha$, 
\begin{equation}\label{2014-01-16+1}
 x_i\sim (1+\alpha)\Gamma(2+\alpha)i^{-2-\alpha}
\qquad
{\text{as}}
\quad
i\to+\infty.
\end{equation}

{\bf Related work}. Results of an empirical study of an evolving coautorhip network
(an affiliation  network, where auhors are declared adjacent if they have a joint publication)
are reported in 
\cite{NewmanMartinBallKarrer2013}. The model considered in the present paper seems 
to be new. 
The idea of such a model has been suggested by Colin Cooper. 
The extra logarithmic factor in (\ref{2014-01-13+1}) indicates that the degree distribution
of the preferred attachment affiliation model has a slightly heavier tail in comparison
to that of the related 'usual' preferential attachment model, see  
 \cite{BollobasRiordan2001}, \cite{DereichMorters}, \cite{Mori2002}. 
 On the other hand,
 affiliation network models, where  power law 
scores  (\ref{2014-01-16+1}) are prescribed to items/attributes independently
of the choices of  vertices, have much heavier tails: the proportion of vertices of degree
$i$  scales as  $i^{-1-\alpha}$ as 
$i\to+\infty$, see   \cite{BloznelisDamarackas}, \cite{BloznelisKaronski}.
An important property of real affiliation networks is that they admit a non-vanishing
clustering coefficient, \cite{Newman+W+S2002}. Clustering characteristics of the preferred
 attachment affiliation 
model will be considered elsewhere.

\medskip

The paper is organized as follows. A heuristic argument explaining (\ref{2014-01-08x1}) 
and (\ref{2014-01-09x}) 
is given in Section 2. A rigorous proof of  (\ref{2014-01-08x1}), (\ref{2014-01-13+1}) and
(\ref{2014-01-09x})   is given in Section 3.

\section{Heuristic}

We start with explaining formula (\ref{2014-01-09x}).
Given $n\ge 1$ and $w\in W_n$, we denote by $s_n(w)$ the score of $w$ after the $n$th step.
By $X^{(n)}_i$ we denote the number of bins $w\in W_n$ of score $s_n(w)=i$. We put $X^{(0)}_1=l$ and $X^{(0)}_i=0$, for $i\ge 2$.

Assume for a moment that for each $i$ the 
ratios $X^{(n)}_i/(nk)$ converge to some limit, say 
${\bar x}_i$, as $n\to+\infty$. So that for large $n$  we have $X^{(n)}_{i}\approx {\bar x}_ink$.
Then from the relations describing approximate behaviour of the  numbers $X^{(n)}_{i}$,
\begin{eqnarray}
\nonumber
 X^{(n+1)}_1
&
\approx
&
(X^{(n)}_1+k)(1-p_{n+1,1}),
\\
\nonumber
X^{(n+1)}_2
&
\approx
&
X^{(n)}_2(1-p_{n+1,2})+(X^{(n)}_{1}+k)p_{n+1,1},
\\
\nonumber
X^{(n+1)}_i
&
\approx
&
X^{(n)}_i(1-p_{n+1,i})+X^{(n)}_{i-1}p_{n+1,i-1},
\qquad
i=3,4,\dots, 
\end{eqnarray}
we obtain, by neglecting $O(n^{-1})$ terms, the equations
\begin{eqnarray}\nonumber
 {\bar x}_1(n+1)k
&=&
({\bar x}_1nk+k)\left(1-\frac{1}{n}\frac{1}{1+\alpha}\right),
\\
\nonumber
{\bar x}_i(n+1)k
&=&
{\bar x}_ink
\left(
1-\frac{1}{n}\frac{i}{1+\alpha}
\right)
+
{\bar x}_{i-1}k\frac{i-1}{1+\alpha},
\quad
 i\ge 2.
\end{eqnarray}
Solving these equations we arrive to the sequence $\{x_i\}_{i\ge 1}$ given by
formula (\ref{2014-01-08x2}).
We remark that $\{x_i\}_{i\ge 1}$ is a sequence of probabilities having a 
finite first moment. More precisely, we have
\begin{equation}\label{2014-01-09+1}
 \sum_{i\ge 1}x_i=1,
\qquad
\sum_{i\ge 1}ix_i=1+\alpha^{-1}.
\end{equation}
In particular, the common probability distribution of random variables $T_i$ is well defined.
We note that identities (\ref{2014-01-09+1}) are simple consequences of the well known
 properties
of the Gamma function and hypergeometric series (formulas (6.1.46), (15.1.20) of
\cite{AS}).

Next we explain (\ref{2014-01-08x1}).
 We call $w\in W_n$ and $v\in V_n$ related whenever
$w$ contains a ball produced by $v$. The number of balls produced by $v$ is 
called the activity 
of $v$. A vertex  $v\in V_n$ is called  regular in  $G_n$ 
if every vertex adjacent to $v$ 
in $G_n$ shares with $v$ a single bin. 
Introduce  event
${\cal V}_{i,r}=\{v_{n+1}$ has activity $r$,  
it has degree $i$ in $G_{n+1}$, and it is a regular vertex of $G_{n+1}\}$ and 
let $q^{(n)}_{i,r}$ denote its probability.
We observe that, given $X^{(n)}_1,X^{(n)}_2,\dots$, 
the conditional probability of the event ${\cal V}_{i,r}$
 is
\begin{equation}\label{2014-01-10+2}
q^{(n)}_r
\sum_{
\begin{subarray}{c}
u_1+c\dots+u_{i+1}=r,\\
1u_2+2u_3+\cdots+iu_{i+1}=i 
\end{subarray}
}
\frac{(X^{(n)}_1+k)_{u_1}(X^{(n)}_2)_{u_2}\cdots (X^{(n)}_{i+1})_{u_{i+1}}}
{(X^{(n)}_1+X^{(n)}_2+\cdots)_{r}}\frac{r!}{u_1!\cdots u_{i+1}!}+o(1).
\end{equation}
Here we use notation $(x)_u=x(x-1)\cdots(x-u+1)$, 
$u_s$ counts those bins $w\in W_{n+1}$ of score $s_n(w)=s$ that have 
received a ball from $v_{n+1}$, and  $q^{(n)}_r$ is the conditional probability, 
given $X^{(n)}_1,X^{(n)}_2,\dots$, of the event that $v_{n+1}$ has produced
$r$ balls. The remainder $o(1)$ accounts for the pobability that $v_{n+1}$ is not a 
regular vertex 
of $G_{n+1}$.

Now, using the approximations 
$X^{(n)}_i\approx x_i nk$, $i\ge 1$,
and identities (\ref{2014-01-09+1}) we, firstly, approximate the first fraction of 
(\ref{2014-01-10+2})
by $x_1^{u_1}\cdots x_{i+1}^{u_{i+1}}$ and, secondly, 
we approximate $q^{(n)}_r$ by the Poisson probability 
$e^{-\lambda}\lambda^r/r!$. We obtain that
\begin{displaymath}
q^{(n)}_{i,r}
\approx
e^{-\lambda}\frac{\lambda^r}{r!}
\sum_{
\begin{subarray}{c}
u_1+c\dots+u_{i+1}=r,\\
1u_2+2u_3+\cdots+iu_{i+1}=i 
\end{subarray}
}
x_1^{u_1}x_2^{u_2}\cdots x_{i+1}^{u_{i+1}}\frac{r!}{u_1!\cdots u_{i+1}!}
=
:c_{i,r}.
\end{displaymath}

Furthermore, we call $v\in V_n$  an $[i,r]$ vertex if its activity is $r$ and 
its degree in $G_n$ is 
$d(v)=i$. By $s_{i,r}(v)$ we denote the (current) number of balls contained in the 
bins related to an $[i,r]$ vertex $v$ of $G_n$. 
We note that any regular $[i,r]$ vertex $v$ of $G_n$ has 
$s_{i,r}(v)=i+2r=:s_{i,r}$.
Moreover,
the probability that $v_{n+1}$ sends a ball to a bin related to such a vertex $v$ is 
$s_{i,r}p_{n+1,1}+O(n^{-2})$.

Let $Y^{(n)}_{i}$ denote the number of regular vertices of $G_n$ of degree $d(v)=i$, 
and let 
$Y^{(n)}_{i,r}$ denote the number of regular $[i,r]$ vertices  of $G_n$.
Assume for a moment that  
for each $i$ and $r$ the ratios $Y^{(n)}_i/n$ converge to some limit, say
${\bar y}_i$, and  $Y^{(n)}_{i,r}/n$ converge to some limit, 
say ${\bar y}_{i,r}$, as $n\to+\infty$. 
So that for large $n$ we have $Y^{(n)}_i\approx {\bar y}_in$ and
$Y^{(n)}_{i,r}\approx {\bar y}_{i,r}n$.
Invoking these approximations in the relations describing 
approximate behaviour of numbers $Y^{(n)}_{i,r}$,
\begin{eqnarray}
\nonumber
 Y^{(n+1)}_{0,0}
&
\approx
&
Y^{(n)}_{0,0}+q^{(n)}_{0,0},
\\
\nonumber
Y^{(n+1)}_{0,r}
&
\approx
&
Y^{(n)}_{0,r}(1-s_{0,r}p_{n+1,1})+q^{(n)}_{0,r},
\qquad
r\ge 1, 
\\
\nonumber
Y^{(n+1)}_{i,r}
&
\approx
&
Y^{(n)}_{i,r}(1-s_{i,r}p_{n+1,1})
+
Y^{(n)}_{i-1,r}s_{i-1,r}p_{n+1,1}
+
q^{(n)}_{i,r},
\qquad
i,r\ge 1.
\end{eqnarray}
we obtain, by neglecting $O(n^{-1})$ terms and using the approximation
$q^{(n)}_{i,r}\approx c_{i,r}$, the equations
\begin{eqnarray}\nonumber
 {\bar y}_{0,0}
&=&
c_{0,0},
\\
{\bar y}_{0,r}
&=&
\frac{1+\alpha}{1+\alpha+2r}c_{0,r},
\qquad
r\ge 1,
\\
{\bar y}_{i,r}
&=&
\frac{2r+i-1}{1+\alpha+2r+i}{\bar y}_{i-1,r}
+
\frac{1+\alpha}{1+\alpha+2r+i}c_{i,r},
\qquad
i,r\ge 1.
\end{eqnarray}
Solving these equations we arrive to the sequence $\{y_{0,0}, \, y_{i,r}, i\ge 0, r\ge 1\}$
 given by
the formulas
\begin{eqnarray}
 y_{0,0}
&=&
c_{0,0},
\\
\label{2014-01-13X1}
y_{i,r}
&=&
(1+\alpha)
\sum_{j=0}^i
\frac{(2r+i-1)_{i-j}}{(1+\alpha+2r+i)_{i-j+1}}c_{j,r}.
\end{eqnarray}
Next we use the identity 
$c_{j,r}=\PP(Z=j,\Lambda=r)=\E {\mathbb I}_{\{\Lambda=r\}}{\mathbb I}_{\{Z=j\}}$
and write (\ref{2014-01-13X1}) in the form
\begin{eqnarray}
\nonumber
y_{i,r}
&=&
(1+\alpha)
\E {\mathbb I}_{\{\Lambda=r\}}{\mathbb I}_{\{Z\le i\}}
\frac{(2\Lambda+i-1)_{i-Z}}{(1+\alpha+2\Lambda+i)_{i-Z+1}}.
\end{eqnarray}
Hence we obtain, for $i\ge 1$,
\begin{eqnarray}
\nonumber
y_i
&=&
\sum_{r\ge 1}y_{i,r}
\\
\nonumber
&=&
(1+\alpha)
\E 
{\mathbb I}_{\{\Lambda\ge 1 \}}
{\mathbb I}_{\{Z\le i\}}
\frac{(i+2\Lambda-1)_{i-Z}}{(i+2\Lambda+\alpha+1)_{i-Z+1}}
\\
\nonumber
&=&
(1+\alpha)
\E 
{\mathbb I}_{\{\Lambda\ge 1 \}}
{\mathbb I}_{\{Z\le i\}}
\frac{\Gamma(i+2\Lambda)}{\Gamma(i+2\Lambda+\alpha+2)}
\frac{\Gamma(Z+2\Lambda+\alpha+1)}{\Gamma(Z+2\Lambda)},
\end{eqnarray}
and 
\begin{displaymath}
 y_0
=
\sum_{r\ge 0}y_{0,r}
=
\PP(\Lambda=0)+
\E 
{\mathbb I}_{\{\Lambda\ge 1 \}}
{\mathbb I}_{\{Z= 0\}}
\frac{1+\alpha}{2\Lambda+\alpha+1}
=
\E 
{\mathbb I}_{\{Z= 0\}}
\frac{1+\alpha}{2\Lambda+\alpha+1}.
\end{displaymath}
We remark that these  identities imply (\ref{2014-01-08x1}), because for every $i\ge 0$,
the number of non 
regular vertices of $G_n$ of degree $i$  can be shown to be negligible.

%,,,,,,,,,,,,,,,,,,,,,,,,,,,,,,,,,,,,,,,,,,,,,,,,,,,,,,,,,,,,

\section{Appendix}

Let ${\tilde Y}^{(n)}_{i,r}$ denote the number of non regular $[i,r]$
vertices of $G_n$.

\begin{proof}[Proof of (\ref{2014-01-08x1}), (\ref{2014-01-13+1}), (\ref{2014-01-09x})]
Let us prove (\ref{2014-01-08x1}).
Let $S_n$ denote the total number of balls in the network after the $n$-th step.
A simple induction argument shows that $\E S_n=l+nk+n\lambda$.
Let $\mathring{Y}^{(n)}_r$ denote the number of vertices $v\in V_n$ with 
activity at least $r$. We observe that for any $0<\varepsilon<1$ 
\begin{equation}\label{2014-01-11+++1}
\sup_n\PP(n^{-1}\mathring{Y}^{(n)}_r > \varepsilon)\to 0
\end{equation}
as $r\to \infty$. Indeed,  vertices of $V_n$ with activity at least
$r$ contribute at least $r\mathring{Y}^{(n)}_r$ balls to $S_n$. Hence,
$\mathring{Y}^{(n)}_r\le r^{-1}S_n$ and we obtain (\ref{2014-01-11+++1}), 
by Markov's inequality.
Now (\ref{2014-01-08x1}) follows from 
(\ref{2014-01-11+++1}) and the fact 
that $n^{-1}{\tilde Y}^{(n)}_{i,r}\to 0$ and
$n^{-1}Y^{(n)}_{i,r}\to y_{i,r}$ in probability as $n\to+\infty$
for $(i,r)=(0,0)$ and $i\ge 0$, $r\ge 1$. This fact follows from  Lemma \ref{y0}:  
we have $n^{-1}\E {\tilde Y}^{(n)}_{i,r}\to 0$, $n^{-1}\E Y^{(n)}_{i,r}\to y_{i,r}$ and
$\Var (n^{-1}Y^{(n)}_{i,r})\to 0$.

Relation (\ref{2014-01-09x}) follows from (\ref{2014-01-13++1}): we have
$(nk)^{-1}\E X^{(n)}_{i}\to x_{i}$ and
$\Var ((nk)^{-1}X^{(n)}_{i})\to 0$.
%of Lemma \ref{x00}, 
%by Chebyshev's inequality .

Let us prove (\ref{2014-01-13+1}). Since the Poisson random variable
$\Lambda$ is highly concentrated around its (finite) mean,  we can approximate 
 with a  high probability for $i,z\to+\infty$
\begin{displaymath}
 \frac{\Gamma(i+2\Lambda)}{\Gamma(i+2\Lambda+\alpha+2)}\approx i^{-2-\alpha},
\qquad
\frac{\Gamma(z+2\Lambda+\alpha+1)}{\Gamma(z+2\Lambda)}\approx z^{1+\alpha}.
\end{displaymath}
Hence, we obtain $y_i\sim (1+\alpha)\E Z^{1+\alpha}{\mathbb I}_{\{Z\le i\}}$ as $i\to+\infty$.
Next, to the randomly stopped sum $Z$ of independent random variables $T_i$
we apply the relation $\PP(Z>t)\sim \PP(T_1>t)\E \Lambda$, \cite{Foss}. We obtain
$\PP(Z>t)\sim \lambda \Gamma(2+\alpha)t^{-1-\alpha}$. The latter
relation implies 
$\E Z^{1+\alpha}{\mathbb I}_{\{Z\le i\}}\sim \lambda(1+\alpha)\Gamma(2+\alpha)\ln i$ 
for $i+\infty$. We have arrived to (\ref{2014-01-13+1}).
\end{proof}

The remaining part of the section contains auxiliary lemmas.

We write for short  $p_{n+1,s}=p_s=s\varkappa_n$, where 
\begin{equation}\label{2014-01-18++1}
\varkappa_n
=
p_{n+1,1}
=
\frac{1}{n}\frac{1}{1+\alpha}
\left(
1-\frac{1}{n}\frac{\alpha+\beta}{1+\alpha+n^{-1}\alpha+n^{-1}\beta}
\right),
\qquad
\beta:=\frac{l}{\lambda}.
\end{equation}
Denote
\begin{eqnarray}\nonumber
x^{(n)}_i
&=&
(nk)^{-1}\E X^{(n)}_i,
\qquad
y^{(n)}_{i,r}
=
n^{-1}\E Y^{(n)}_{i,r},
\qquad
{\tilde y}^{(n)}_{i,r}
=
n^{-1}\E {\tilde Y}^{(n)}_{i,r},
\\
\nonumber
h_{i,j}^{(n)}
&=&
(nk)^{-1}
\left(
\E X_i^{(n)}X_j^{(n)}-\E X_i^{(n)}\E X_j^{(n)}
\right),
\quad \,
g^{(n)}_{i,j;r}
=
n^{-2}
\left(
\E Y^{(n)}_{i,r}Y^{(n)}_{j,r}-\E Y^{(n)}_{i,r}\E Y^{(n)}_{j,r}
\right).
\end{eqnarray}

\begin{lem} \label{x00} For any $i,j\ge 1$ we have as $n\to+\infty$
\begin{equation}\label{2014-01-13++1}
 x^{(n)}_i=x_i+O(n^{-1}),
\qquad
(nk)^{-2}\E X^{(n)}_iX^{(n)}_j=x_ix_j+O(n^{-1}).
\end{equation}
Moreover,  the finite limits 
\begin{equation}\label{limitHij}
 h_{i,j}=\lim_nh_{i,j}^{(n)},
\qquad
i,j\ge 1,
\end{equation}
exist and can be calculated using the recursive relations
\begin{eqnarray}\label{Hii}
&&
h_{i,i}=\frac{2(i-1)h_{i,i-1}+ix_i+(i-1)x_{i-1}}{i+i+1+\alpha},
\\
\label{Hii+1}
&&
h_{i,i+1}=
\frac{(i-1)h_{i-1,i+1}+ih_{i,i}-ix_i}{i+(i+1)+1+\alpha},
\\
\label{Hir}
&&
 h_{i,r}=\frac{(i-1)h_{i-1,r}+(r-1)h_{i,r-1}}{i+r+1+\alpha},
\qquad
r\ge i+2.
\end{eqnarray}
In particular, we have for every $i,j\ge 1$,
\begin{equation}\label{zijh}
 (nk)^{-2}\E X^{(n)}_iX^{(n)}_j=x_i^{(n)}x_j^{(n)}+h_{i,j}(nk)^{-1}+o(n^{-1}).
\end{equation}
Here we use notation $x_0\equiv 0$ and $h_{i,j}\equiv 0$, for $\min\{i,j\}=0$.
\end{lem}

\begin{proof}[Proof of Lemma \ref{x00}] 
Let us prove the first relation of (\ref{2014-01-13++1}). The identities
\begin{eqnarray}\nonumber
\E X^{(n+1)}_1
&=&
(1-p_1)\E (X^{(n)}_1+k),
\\
\nonumber
\E X^{(n+1)}_2
&=&
(1-p_2)\E X_2^{(n)}+p_1\E (X_1^{(n)}+k),
\\
\nonumber
%\label{2014-01-13++}
\E X^{(n+1)}_i
&=&
(1-p_i)\E X_i^{(n)}+p_{i-1}\E X_{i-1}^{(n)},
\quad
i\ge 1,
\end{eqnarray}
imply
\begin{eqnarray}\label{2014-01-13++3}
x^{(n+1)}_1
&=&
x^{(n)}_1(1-n^{-1}-p_1)+n^{-1}+O(n^{-2}), 
\\
\label{2014-01-13++4}
x^{(n+1)}_i
&=&
x^{(n)}_i(1-n^{-1}-p_i)+x^{(n)}_{i-1}p_{i-1}+O(n^{-2}),
\quad
i\ge 1. 
\end{eqnarray}
Relation (\ref{2014-01-13++3}) combined with Lemma \ref{L1} implies
$x^{(n)}_1=x_1+O(n^{-1})$.
For $i\ge 2$ we proceed recursively: using the fact that  $x^{(n)}_{i-1}=x_{i-1}+O(n^{-1})$
we conclude from (\ref{2014-01-13++4}) by Lemma \ref{L1} that 
$x^{(n)}_{i}=x_{i}+O(n^{-1})$.

Next, we observe that the second relation of (\ref{2014-01-13++1}) follows 
from (\ref{zijh}). Furthermore, (\ref{zijh}) 
follows from (\ref{Hii}), (\ref{Hii+1}) and (\ref{Hir}). Hence we only need to
prove (\ref{Hii}), (\ref{Hii+1}) and (\ref{Hir}).

For convenience we write $h^{(n)}_{i,j}\equiv 0$, for $\min\{i,j\}=0$. We also 
put $x_0^{(n)}\equiv 0$.  Clearly, 
$h^{(n)}_{i,j}=h^{(n)}_{j,i}$ for $i,j\ge 0$.

Let us prove (\ref{Hii}). 
A  straightforward calculation shows that
\begin{eqnarray}
%h_{1,1}^{(n+1)}
%\frac{n+1}{n}
%&=&
%h_{1,1}^{(n)}(1-p_1)^2
%+
%x_1^{(n)}(p_1-p_1^2)
%+
%O(n^{-2}),
%\\
\label{hii}
h_{i,i}^{(n+1)}
\frac{n+1}{n}
&=&
h_{i,i}^{(n)}(1-p_i)^2
+
h_{i,i-1}^{(n)}2(1-p_i)p_{i-1}
\\
\nonumber
&+&
x_{i-1}^{(n)}(p_{i-1}-p_{i-1}^2)
+
x_{i}^{(n)}(p_i-p_i^2)\frac{n}{n+1}
+O(n^{-2}),
\\
\label{hii+1}
h_{i,i+1}^{(n+1)}
\frac{n+1}{n}
&=&
h_{i,i+1}^{(n)}(1-p_i)(1-p_{i+1})
+
h_{i-1,i+1}^{(n)}p_{i-1}(1-p_{i+1})
\\
\nonumber
&+&
h_{i,i}^{(n)}p_{i}(1-p_{i})
-
x_i^{(n)}p_i(1-p_i)
+
O(n^{-2}),
\end{eqnarray}
and, for  $r\ge 2+i$, 
\begin{eqnarray}
% \label{h1r}
%h_{1,r}^{(n+1)}\frac{n+1}{n}
%&=&
%h_{1,r}^{(n)}(1-p_1)(1-p_r)
%+
%h_{1,r-1}^{(n)}p_{r-1}(1-p_1),
%\\
\label{hir}
h_{i,r}^{(n+1)}\frac{n+1}{n}
&=&
h_{i,r}^{(n)}(1-p_i)(1-p_r)
+
h_{i-1,r}^{(n)}p_{i-1}(1-p_r)
\\
\nonumber
&+&
h_{i,r-1}^{(n)}p_{r-1}(1-p_i)+O(n^{-2}).
\end{eqnarray}

We note that (\ref{hii}) and
Lemma \ref{L1} imply that the sequence $\{h_{1,1}^{(n)}\}_{n\ge 1}$ converges to 
$h_{1,1}$ defined by (\ref{Hii}). Furthermore, using the fact that (\ref{limitHij}) holds
for $i=j=1$
we obtain from (\ref{hii+1}) and Lemma \ref{L1} that $\{h_{1,2}^{(n)}\}_{n\ge 1}$ 
converges to 
$h_{1,2}$ defined by (\ref{Hii+1}). Next, for $i=1$ and $r=3,4,\dots$, 
we proceed recursively: 
using  (\ref{hir}) and Lemma \ref{L1} we establish
(\ref{limitHij}), with $h_{ir}$ given by (\ref{Hir}). In this way we prove the lemma for
$i=1$ and $r\ge i$.

The case $i=2$, $r\ge i$ is treated similarly. For $i=r=2$ we apply (\ref{hii}) and
Lemma \ref{L1}. For $i=2$ and $r=3$ we apply (\ref{hii+1}) and
Lemma \ref{L1}. Finally, for $i=2$ and $r\ge i+2$ we apply (\ref{Hir}) and Lemma \ref{L1}.
 
Next we proceed recursively and prove the lemma 
for $\{(i, r), r=i, r=i+1, r=i+2,\dots\}$, $i=3,4,\dots$.
\end{proof}

\begin{lem} \label{y0} Let $i,j=0,1,\dots$ and $r=1,2\dots$. 
We have as $n\to+\infty$
\begin{equation}\label{12-23+6}
y^{(n)}_{i,r}\to y_{i,r},
\qquad
 g^{(n)}_{i,j;r} \to 0,
\qquad
{\tilde y}^{(n)}_{i,r}\to 0.  
\end{equation}
(\ref{12-23+6}) remains valid for $i=j=r=0$. 
\end{lem}
\begin{proof}[Proof of Lemma  \ref{y0}]

For $i,j,r\ge 1$ we show that
\begin{eqnarray}\label{2014-01-25+1}
{\tilde y}^{(n)}_{0,0}
&\equiv&
{\tilde y}^{(n)}_{0,r}
\equiv
{\tilde y}^{(n)}_{i,1},
\\
\label{2014-01-25+2}
{\tilde y}^{(n+1)}_{i+1,r}
&\le&
{\tilde y}^{(n)}_{i+1,r}(1-n^{-1})
+
{\tilde y}^{(n)}_{i,r}s_{ir,r}\varkappa_n
%+
%y^{(n)}_{i,r}(s_{i,r}\varkappa_n)^2
+
o(n^{-1}),
\\
\label{2014-01-07x1}
 y^{(n+1)}_{0,0}
&=&
(1-n^{-1})y^{(n)}_{0,0}+n^{-1}c_{0,0}+o(n^{-1}),
\\
\label{2014-01-07x2}
 y^{(n+1)}_{0,r}
&=&
\left(1-n^{-1}-s_{0,r}\varkappa_n\right)y^{(n)}_{0,r}
+
n^{-1}c_{0,r}+o(n^{-1}),
\\
\label{2014-01-07x3}
y^{(n+1)}_{i,r}
&=&
\left(1-n^{-1}-s_{i,r}\varkappa_n)\right)
y^{(n)}_{i,r}
+
s_{i-1,r}\varkappa_n
y^{(n)}_{i-1,r}
+
n^{-1}c_{i,r}+o(n^{-1}),
\end{eqnarray}
and
\begin{eqnarray}
\label{2014-01-07x4}
 g^{(n+1)}_{0,0;0}
&=&
(1-2n^{-1})g^{(n)}_{0,0;0} 
+
o(n^{-1}),
\\
\label{2014-01-07x5}
g^{(n+1)}_{0,0;r}
&=&
(1-2n^{-1}-2s_{0,r}\varkappa_n)
g^{(n)}_{0,0;r}
+
o(n^{-1}),
\end{eqnarray}
\begin{eqnarray}
\label{2014-01-07x6}
g^{(n+1)}_{0,j;r}
&=&
(1-2n^{-1}-(s_{j,r}+s_{0,r})\varkappa_n)
g^{(n)}_{0,j;r}
+
s_{j-1,r}\varkappa_n
g^{(n)}_{0,j-1;r}
+o(n^{-1}),
\\
\label{2014-01-07x7}
g^{(n+1)}_{i,j;r}
&=&
(1-2n^{-1}-(s_{i,r}+s_{j,r})\varkappa_n)
g^{(n)}_{i,j;r}
+
s_{i-1,r}\varkappa_n
g^{(n)}_{i-1,j;r}
+
s_{j-1,r}\varkappa_n
g^{(n)}_{i,j-1;r}
\\
\nonumber
&+&
o(n^{-1}).
\end{eqnarray}
The proof of (\ref{2014-01-25+1})-(\ref{2014-01-07x7}) is technical.
We refer the reader to the extended version of the paper \cite{BloznelisGotzeExtended}
for details.
Here we prove that 
 (\ref{2014-01-25+1})-(\ref{2014-01-07x7}) imply (\ref{12-23+6}).

Let us prove the third relation of  (\ref{12-23+6}). For $i=0$, and for $r=0,1$ 
the relation follows from 
(\ref{2014-01-25+1}). Next, for any fixed $r\ge 2$ we proceed 
recursively: from (\ref{2014-01-25+2}) combined with the fact that
${\tilde y}^{(n)}_{i,r}\to 0$  
we conclude by Lemma \ref{L1}
that ${\tilde y}^{(n)}_{i+1,r}\to 0$.

Let us prove the first and second relation of (\ref{12-23+6}). 
Firstly, combining (\ref{2014-01-07x1}) (respectively (\ref{2014-01-07x4})) 
with Lemma \ref{L1} we obtain the first (respectively second) relation
of (\ref{12-23+6}), for $i=j=r=0$.
Secondly, combining (\ref{2014-01-07x2}) (respectively (\ref{2014-01-07x5})) 
with Lemma \ref{L1}
we obtain the first
(respectively second) relation
of (\ref{12-23+6}), for $i=j=0$, $r\ge 1$.

Now  we prove the first relation of  (\ref{12-23+6}) for  $i\ge 1$ and $r\ge 1$.
We fix $r$ and proceed recursively: from the fact that $y^{(n)}_{i-1,r}\to y_{i-1,r}$
and relation (\ref{2014-01-07x3}) we conclude by Lemma \ref{L1} that
$y^{(n)}_{i,r}\to y_{i,r}$.

Next, we prove  the second relation of (\ref{12-23+6}) for $r\ge 1$ and $i+j\ge 1$.
We fix $r$ and proceed recursively in $i$ and $j$.

For $i=0$ and $j\ge 1$ we proceed as follows: from the fact that 
$g^{(n)}_{0,j-1;r}\to 0$
and relation (\ref{2014-01-07x6}) we conclude by Lemma \ref{L1} that
$g^{(n)}_{0,j;r}\to 0$.
In this way we prove the second relation of (\ref{12-23+6}) for $(i,j)$ such that
$i=0$ and $j\ge 1$.

Now, consider indices $i=1$ and $j\ge 1$. From the fact that 
$g^{(n)}_{1,j-1;r}\to 0$
and relation (\ref{2014-01-07x6}) we conclude by Lemma \ref{L1} that
$g^{(n)}_{1,j;r}\to 0$.
In this way we prove the second relation of (\ref{12-23+6}) for $(i,j)$ such that
$i=1$ and $j\ge 2$.

Proceeding similarly we establish the second relation of (\ref{12-23+6})
for $\{(i,i), (i,i+1), (i,i+2),\dots\}$, $i=2,3,\dots$.
\end{proof}

\begin{lem}\label{L1} Let $b,h\in R$.
Let $\{b_n\}_{n\ge 1}$ be a
real sequence  converging to $b$ and assume that
the series $\sum_{n\ge 1}n^{-1}|b_n-b|$ converges. 
Let $\{h_n\}_{n\ge 1}$ be a
real sequence converging to $h$.
Let $\{a_n\}_{n\ge 1}$ be a real sequence satisfying the recurrence relation 
\begin{equation}\label{recurrence_a}
 a_{n+1}=a_n(1-n^{-1}b_n)+n^{-1}h_n,
\qquad
n\ge 1.
\end{equation}
 For  $b> 0$ we have  $a_n\to hb^{-1}$. Suppose, in addition, that
$b_n-b=O(n^{-1})$, $h_n-h=O(n^{-1})$. Then for $b\not=1$  we have
$a_n-hb^{-1}=O(n^{-1\wedge b})$, and for $b=1$ we have $a_n-hb^{-1}=O(n^{-1}\ln n)$.

Let ${\tilde b}\ge 0$.
Let $\{{\tilde a}_n\}_{n\ge 1}$, $\{{\tilde b}_n\}_{n\ge 1}$, $\{{\tilde h}_n\}_{n\ge 1}$
be non negative sequences such that ${\tilde b}_n\to {\tilde b}$, ${\tilde h}_n\to 0$
and  $\{{\tilde a}_n\}_{n\ge 1}$ satisfies the inequality
\begin{equation}\nonumber
 {\tilde a}_{n+1}
\le 
{\tilde a}_n(1-n^{-1}{\tilde b}_n)+n^{-1}{\tilde h}_n,
\qquad
n\ge 1.
\end{equation}
 Assume that the series $\sum_{n\ge 1}n^{-1}|{\tilde b}_n-{\tilde b}|$ converges. 
Then $\{{\tilde a}_n\}_{n\ge 1}$ converges to $0$.
\end{lem}

The proof is straightforward, see \cite{BloznelisGotzeExtended} for details.

{\bf Acknowledgement}. M. Bloznelis thanks Katarzyna Rybarczyk for discussion.
 The work of M. Bloznelis was supported by the Research Council
of Lithuania
grant MIP-067/2013 and by the SFB 701  grant at Bielefeld university.

%\end{document}

\end{document}